\documentclass[twocolumn,pra,showpacs,amsmath,amssymb]{revtex4}

\usepackage{amsmath,amsfonts,amsthm,graphicx,color}

\usepackage{float}
\usepackage{hyperref}

\begin{document}

\title{Generalized Clauser-Horne-Shimony-Holt inequalities maximally violated by higher dimensional systems}
\author{T. V\'ertesi}
\email{tvertesi@dtp.atomki.hu}
\affiliation{Institute of Nuclear Research of the Hungarian Academy of Sciences\\
H-4001 Debrecen, P.O. Box 51, Hungary}
\author{K.F. P\'al}
\email{kfpal@atomki.hu}
\affiliation{Institute of Nuclear Research of the Hungarian Academy of Sciences\\
H-4001 Debrecen, P.O. Box 51, Hungary}

\def\CC{\mathbb{C}}
\def\RR{\mathbb{R}}
\def\one{\leavevmode\hbox{\small1\normalsize\kern-.33em1}}
\newcommand*{\tr}{\mathsf{Tr}}
\newtheorem{theorem}{Theorem}[section]
\newtheorem{lemma}[theorem]{Lemma}

\date{\today}

\begin{abstract}
Imagine two parties, Alice and Bob who share an entangled quantum state.
A well-established result that if Alice performs two-outcome measurement on the portion of the state in her possession and Bob does likewise, they are able to produce correlations that cannot be reproduced by any classical theory. The allowed classical correlations can be expressed quantitatively by the Bell inequalities.
Here we propose new families of Bell inequalities, as a generalization of the Clauser-Horne-Shimony-Holt (CHSH) inequality and show that the maximum violation of these Bell inequalities allowed by quantum theory can not be attained by a bipartite quantum system having support on a qubit at each site.
\end{abstract}

\maketitle

\section{Introduction}\label{intro}

Bell inequalities are strict bounds on certain combinations of probabilities and correlation functions for measurements on multipartite systems \cite{Bell64,CHSH69}. These bounds apply for any local realistic theory. In the two-party two-outcome measurement scenario, the case we restrict our attention, Alice performs one of her $m_A$ measurements and Bob performs one of his $m_B$ measurements and then output respectively one of $k_A$ and $k_B$ different outcomes.

In the simplest nontrivial case with two measurement settings and two outcomes per party, there is (up to symmetries)
one nontrivial Bell inequality, the Clauser-Horne-Shimony-Holt (CHSH) inequality \cite{CHSH69}.
There exist generalizations of this CHSH inequality in various directions such as for arbitrary number of measurement
settings (e.g., \cite{BC90,Gisin99,PS01,CG04,AIIS05}), outcomes (e.g., \cite{CGLMP,KKCZO02,BG03,AGG05,CG04,Gisin07}) and for many parties as well (e.g., \cite{Mermin90,BK93,ZB02,LPZB04,SLK06}).
Here we shall focus on the two-party scenario with $m_A>2$ and $m_B>2$ measurement settings having binary outcomes, i.e., the case when $k_A=k_B=2$.

According to quantum mechanics composite systems can be entangled and may not obey a
local realistic description. The nonlocal nature appears evidently in the fact that
entangled states allow for violation of Bell inequalities. For instance, the singlet state
of two spin-$1/2$ particles shared by Alice and Bob violates the CHSH inequality by a multiplicative factor $\sqrt 2$,
but as Tsirelson showed \cite{Tsirelson80} this is the maximum amount of violation attainable on the basis of quantum mechanics. That is, by increasing the size of the local Hilbert space on Alice and Bob's side would not give any advantage in the violation of the CHSH inequalities. Then we may inquire whether there exist two-outcome two-party Bell inequalities at all which are maximally violated by higher than two-dimensional systems.

On one hand, it has been shown that both the generalized CHSH-type inequality for arbitrary $m$ settings of Braunstein and Caves \cite{BC90} and both the inequality of Gisin \cite{Gisin99} can be maximally violated by the use of a maximally entangled pair of qubits. The proof regarding the former inequality was provided by Wehner \cite{Wehner06} by the mean of analytic techniques borrowed from semidefinite programming \cite{BV04}. Further, the fact that the best quantum bound can be achieved by two qubits for Gisin's inequalities, was proved analytically recently by Tsirelson \cite{Tsirelson07}.

On the other hand, one can also consider situations where the number of parties are more than two. However, the theorem presented by Masanes \cite{Masanes05} (and the alternative proof presented in \cite{TV06} by Toner and Verstraete) implies that for an arbitrary number of parties, but for only two measurement settings per party ($m=2$) it suffices to perform projective measurements on systems having support on a qubit by each party, in order to obtain the maximal violation of the corresponding Bell inequalities.

The above results suggest the question, originally posed by Gill \cite{Gill} (see also \cite{Gisin07}):
Can all Bell inequalities with $k$ outcomes be maximally violated by choosing each party's local
Hilbert spaces to be $k$-dimensional and each measurement as a complete von Neumann measurement
(with $k$ orthogonal projectors) on pure states with minimal dimension?

Here we intend to give a definite answer to the first part of the question by providing explicit examples for two-outcome two-party CHSH-type inequalities whose maximal violation is not achieved by qubits. Note, that this question has already been answered in Refs.~\cite{Brunner,Perez} by proving the existence of such two-outcome Bell inequalities in the case of two and three parties, respectively. However, we prove the existence of this kind of inequalities for two parties by explicitly constructing them.

More specifically, in Section~\ref{rep} the connection between the vector construction of Tsirelson and the extremal correlations formed by measurements for systems of local qubits and also for general quantum systems are established. Then in Section~\ref{newfam} we construct CHSH-type Bell inequalities with asymmetric number of measurement settings $m_A>m_B$ on Alice and Bob's side, respectively. It is found numerically in Section~\ref{numeric}, that two of our Bell inequalities with the number of measurement settings $m_A=8$, $m_B=4$ and $m_A=12$, $m_B=4$ can be violated by a quantum state with four-dimensional local Hilbert spaces stronger (with the respective ratios $\sim1.036$ and $\sim1.015$) than if the parties were limited to use only local qubits. In Section~\ref{analytic} we also show analytically for a CHSH-type Bell inequality with $m_A=15$ and $m_B=6$ settings, derived by tailoring one of our family of Bell inequalities, that the ratio in question is definitely bigger than one ($\sim1.012$). We lend analytic results from discrete geometry in order to obtain this result. The paper concludes in Section~\ref{disc} discussing some open questions as well.

\section{Representation of joint correlations with dot products}\label{rep}

In this section we determine joint correlations which can be achieved classically, by the aid of quantum mechanics and by the restricted case that each party possesses a qubit. It is shown that the extremal values of a combination of these correlations can be obtained by the mean of a construction of dot products of Euclidean unit vectors in accordance with Tsirelson's theorem \cite{Tsirelson87,Tsirelson93}.

\subsection{Joint correlations}\label{joint}

\paragraph{Classical correlations:}

Let $a_i, b_j \in \{+1,-1\}$ for indices $1\le i \le m_A$ and $1\le j \le m_B$, where $m_A$ and $m_B$ denotes the number of measurements on Alice and Bob's side, respectively. We can write the expression

\begin{equation}
\label{BM}
  {\cal B}_M =\sum_{i=1}^{m_A}\sum_{j=1}^{m_B}{M_{ij} a_i b_j},
\end{equation}
where $M=(M_{ij})$ is a $m_A \times m_B$ matrix with real entries. A local hidden variables (LHV) model for a bipartite two-outcome scenario can be defined as a protocol \cite{Bell64}: Alice and Bob share a variable $\lambda \in \Lambda$, chosen according to a distribution $q$. Then Alice outputs $\alpha=A(\lambda,i)\in\{+1,-1\}$ and Bob outputs $\beta=B(\lambda,j)\in\{+1,-1\}$, where $i$ and $j$ labels the measurement settings on Alice and Bob's side, respectively. The joint correlation between Alice and Bob's outcomes is defined by averaging over the variable $\lambda$,
\begin{equation}
\langle \alpha_{i} \beta_j \rangle_{\text{LHV}} = \int{d\lambda q(\lambda)A(\lambda,i)B(\lambda,j)}.
\label{jointLHV}
\end{equation}
A generic Bell expression involving joint correlations can be written as
\begin{equation}
\label{BMLHV}
\langle {\cal B}_M \rangle_{\text{LHV}}=\sum_{i,j} M_{ij}\langle \alpha_i\beta_j\rangle_{\text{LHV}}.
\end{equation}
By maximizing $\langle {\cal B}_M \rangle_{\text{LHV}}$ we obtain a bound on the corresponding Bell inequality which has to be satisfied by any LHV model.

\paragraph{Generic quantum correlations:}

In this case we are interested in the quantum value of $\langle {\cal B}_M \rangle$. For this, let us define the quantum measurement model for the bipartite two-outcome measurement scenario: Alice and Bob share a pure state $|\psi\rangle$ in the Hilbert space ${\cal H}$ of finite or countable dimension. The observables $A_1,\ldots,A_{m_A}$ and $B_1,\ldots,B_{m_B}$ corresponding to each party's measurements, have eigenvalues $\pm 1$, that is $A_i^2=B_j^2=\one$. Note, that according to the Theorem~5.4 in Ref.~\cite{CHTW04}, in order to obtain the maximum quantum value of $\langle {\cal B}_M \rangle$, it is sufficient to carry out projective measurements (i.e., observables with the above properties) on system in a pure state. Then the joint correlation of Alice and Bob's measurement results are given by
\begin{equation}
\langle \alpha_i \beta_j \rangle_\text{QM} = \langle \psi| A_i \otimes B_j |\psi\rangle.
\label{jointQM}
\end{equation}

\paragraph{Quantum correlations between qubits:}
The dimension of the local Hilbert spaces has not been specified yet. Here we restrict our attention to the case when Alice and Bob, each has a spin-$1/2$ particle or qubit and the shared pure state is $|\psi\rangle \in \CC^2 \otimes \CC^2$. They each measure their own spin along a direction, specified by the orientation of the corresponding Stern-Gerlach apparatus, given by a unit vector from the sets $\{\vec a_i\}_{i=1}^{m_A} \in \RR^3$ and $\{\vec b_j\}_{j=1}^{m_B} \in \RR^3$. The corresponding observables are
\begin{subequations}
\label{observe}
\begin{align}
A_i &=\vec a_i \cdot \vec\sigma, \label{Ai}\\
B_j &=\vec b_j \cdot \vec\sigma^t,
\label{Bj}
\end{align}
\end{subequations}
where $\vec\sigma$ is the vector of the three Pauli matrices. The transposition (denoted by $t$) in Eq.~(\ref{Bj}) has been applied for later convenience.
The joint correlations arising in this case are denoted by $\langle \alpha_{i} \beta_j \rangle_{\text{2D}}$ and $\langle \alpha_{i} \beta_j \rangle_{\text{3D}}$ depending on whether the measurement directions $\{\vec a_i\}_{i=1}^{m_A}$ and $\{\vec b_j\}_{j=1}^{m_B}$ are chosen from co-planar settings, such as
\begin{subequations}
\label{directions}
\begin{align}
\vec a_i &= (\sin \theta_{ai},0,\cos\theta_{ai}),\\
\vec b_j &= (\sin \theta_{bj},0,\cos\theta_{bj}),
\end{align}
\end{subequations}
or are allowed to point to arbitrary points on the Poincar{\'e} sphere.

\subsection{Extremal values of correlation Bell inequalities}\label{extremal}

Let us define the following expression,
\begin{equation}
\label{BMd}
  \langle{\cal B}_M\rangle_d =\sum_{i=1}^{m_A}\sum_{j=1}^{m_B}{M_{ij}\vec{\vphantom{b} a}_i \vec b_j},
\end{equation}
where the unit vectors $\vec a_1,\ldots,\vec a_{m_A} \in \RR^d$ and $\vec b_1,\ldots,\vec b_{m_B} \in \RR^d$, that is, the vectors are inscribed in the unit sphere $S^{d-1}$. Without loss of generality let us assume $m_A\ge m_B$.

Next we can maximize Eq.~(\ref{BMd}),
\begin{equation}
\label{maxBMd}
\max\langle{\cal B}_M\rangle_{d}=\max_{\vec b_j \in S^{d-1}}\sum_{i=1}^{m_A}\left|\sum_{j=1}^{m_B}M_{ij}\vec b_j\right|
\end{equation}
by choosing
\begin{equation}
\label{ai}
\vec a_i=\frac{\sum_{j=1}^{m_B}M_{ij}\vec b_j}{|\sum_{j=1}^{m_B}M_{ij}\vec b_j|}
\end{equation}
for all $1\le i\le m_A$, that is each $\vec a_i$ is parallel to the linear combination of the $\vec b_j$ vectors it is multiplied with. Here we used the notation $|\vec v|$ for the Euclidean norm of a vector $\vec v\in \RR^d$.

\paragraph{Classical bound:}

Comparing Eq.~(\ref{BMd}) with Eq.~(\ref{BM}) one immediately arrives at
\begin{equation}
\label{BM1}
{\cal B}_M= \langle{\cal B}_M\rangle_1,
\end{equation}
and also one can write
\begin{equation}
\label{maxBMLHV}
\max\langle{\cal B}_M\rangle_\text{LHV}=\max \langle{\cal B}_M\rangle_1,
\end{equation}
by noticing that averaging can only lower the value of $\langle{\cal B}_M\rangle_\text{LHV}$.

\paragraph{Generic quantum bound:}

A generic correlation Bell expression for quantum correlations can be written as
$\langle{\cal B}_\text{M}\rangle_\text{QM}=\sum_{i=1}^{m_A}\sum_{j=1}^{m_B}M_{ij}\langle \alpha_i \beta_j\rangle_\text{QM},$ where $\langle \alpha_i \beta_j\rangle_\text{QM}$ is defined by Eq.~(\ref{jointQM}).
However, by the mean of Lemma~\ref{lemma:first} in the Appendix, one can write $\langle \alpha_i \beta_j\rangle_\text{QM}=\vec{\vphantom{b} a}_i \cdot \vec b_j$, where $\vec a_i, \vec b_j \in \RR^{m_A+m_B}$, no matter how large the dimension of the local Hilbert spaces of Alice and Bob is. Therefore using the above correspondence and Eq.~(\ref{BMd}) we obtain
\begin{equation}
\label{maxBMQM}
\max\langle{\cal B}_M\rangle_\text{QM}=\max\langle{\cal B}_M\rangle_{m_A+m_B}=\max\langle{\cal B}_M\rangle_{m_B}.
\end{equation}
In the last equality we used the fact that at the extremum Alice's each vector $\vec a_i$ in Eq.~(\ref{ai}) lies in the subspace spanned by Bob's vectors $\{\vec b_j\}_{j=1}^{m_B}$.
Let us mention that Lemma~\ref{lemma:second} indicates that the maximum value in Eq.~(\ref{maxBMQM}) achievable by an optimization strategy based on unit vectors can also be implemented by projective quantum measurements.

\paragraph{Qubit bound:}

Let us define $|a_i\rangle=A_i\otimes\one|\psi\rangle$ and $|b_j\rangle=\one\otimes B_j|\psi\rangle$, where the observables $A_i$ and $B_j$ valid for qubits are defined by Eqs.~(\ref{observe}), $|\psi\rangle\in\CC^2 \otimes\CC^2$. Then we have, independently of the state $|\psi\rangle$,
\begin{equation}
\left\|\sum_{j=1}^{m_B}M_{ij}|b_j\rangle\right\|=\left|\sum_{j=1}^{m_B}M_{ij}\vec b_j\right|,
\label{norms}
\end{equation}
where the unit vectors $\vec b_j\in\RR^3$ are given by Eq.~(\ref{Bj}) and the norm of a Hilbert space vector $|v\rangle$ is defined by $\left\||v\rangle\right\|\equiv \sqrt{\langle v|v\rangle}$.
The proof may go as follows: Let $|v\rangle =\one\otimes\vec v\sigma^t|\psi\rangle$, which is a linear map from $\RR^3$ to $\CC^2\otimes\CC^2$ by sending $\vec v$ to $|v\rangle$. Therefore by setting $|v\rangle = \sum_{j=1}^{m_B}M_{ij}|b_j\rangle$ it suffices to prove that $\left\||v\rangle\right\|=\left|\vec v\right|$. However, this formula immediately follows from the chain of equalities $\left\||v\rangle\right\|= \sqrt{\langle v|v\rangle} = \sqrt{\langle\psi|\one\otimes(\vec v\sigma^t)^2|\psi\rangle}=|\vec v|$, where we used the identity, ${\left(\vec v \vec \sigma^t\right)}^2 = |\vec v|^2\one$, valid due to the fact that the Pauli matrices anticommute and square to the identity. The Bell expression for qubits can be written in the form
\begin{equation}
\langle {\cal B}_M\rangle_\text{3D}=\sum_{i=1}^{m_A}\sum_{j=1}^{m_B}M_{ij}\langle \alpha_i \beta_j\rangle_\text{3D}.
\label{BM3D}
\end{equation}
Alternatively, if the measurement directions are confined to the plane one can write 2D instead of 3D in the subscript.
However, $\langle \alpha_i\beta_j\rangle_\text{3D}=\langle a_i|b_j\rangle$ and the formulae $\langle a_i|a_i\rangle=\langle b_j|b_j\rangle=1$ hold, owing to $A_i^2=B_j^2=\one$. Thus, we can further write Eq.~(\ref{BM3D}) to obtain an upper bound on it,
\begin{align}
\langle {\cal B}_M\rangle_\text{3D}=&\sum_{i=1}^{m_A}\langle a_i|\sum_{j=1}^{m_B}M_{ij}|b_j\rangle \nonumber\\ \le &\sum_{i=1}^{m_A}\left\||a_i\rangle\right\|\cdot \left\|\sum_{j=1}^{m_B}M_{ij}|b_j\rangle\right\|=\sum_{i=1}^{m_A}\left|\sum_{j=1}^{m_B}M_{ij}\vec b_j\right|,
\label{innerBM3D}
\end{align}
where we used the Cauchy-Schwarz inequality and then Eq.~(\ref{norms}) to obtain the last member. But the inequality in question can always be saturated by choosing $\vec a_i$ according to Eq.~(\ref{ai})
and by sharing the maximally entangled state $|\psi\rangle=|\Phi^+\rangle=1/\sqrt{2} \sum_{i=1}^{2}{|ii\rangle}$.
From this fact and from Eqs.~(\ref{maxBMd},\ref{BM3D},\ref{innerBM3D}) we gain the following relations
\begin{subequations}
\label{maxBM2D3D}
\begin{align}
\max\langle{\cal B}_M\rangle_\text{2D} &=\max\langle{\cal B}_M\rangle_\text{2},\\
\max\langle{\cal B}_M\rangle_\text{3D} &=\max\langle{\cal B}_M\rangle_\text{3},
\end{align}
\end{subequations}
where 2D signifies that the maximum can be attained by measurements performed on qubits with corresponding states ($|\Phi^+\rangle$) and observables (Eqs.~(\ref{observe},\ref{directions})) which can be written, in an appropriate basis, using only real numbers. In contrast, 3D denotes that the measurement settings can only be expressed using complex numbers.

In summary, in this section it has been shown in agreement with Tsirelson's work \cite{Tsirelson87}, that the highest achievable classical and quantum values can be obtained by maximizing the formula~(\ref{BMd}) with respect to the unit vectors $\vec b_j$ for $d=1$ and $d=m_B$, respectively. On the other hand, when Alice and Bob each possesses a qubit, one has $d=2$ and $d=3$ in formula~(\ref{BMd}) depending on whether the measurement settings can be taken from co-planar settings or not in order to obtain a maximal value for the Bell expression~(\ref{BM3D}). In the case of a qubit-qubit system the maximum quantum violation of any correlation Bell inequality~(\ref{BMLHV}) is achieved by a maximally entangled state.

\section{New families of Bell inequalities}\label{newfam}

Let us specify the $M_{ij}$ entries of the matrix $M$ in Eq.~(\ref{BM}) with the following inequalities
\begin{subequations}
\label{Bxyz}
\begin{align}
{\cal B}_{Xn} \equiv& \sum_{k_1,k_2,\ldots,k_{n-1}\in\{0,1\}}{a_{k_1 k_2 \dots k_{n-1}}}\nonumber\\ &\times \left((-1)^{k_1}b_1+(-1)^{k_2}b_2+ \dots +(-1)^{k_{n-1}}b_{n-1}+b_n\right) \nonumber \\ \le &\sum_{i=0}^n {\binom{n}{i}|n-2i|}= \left\lfloor\frac{n+1}{2}\right\rfloor \binom{n}{\left\lfloor\frac{n+1}{2}\right\rfloor} \label{Bx}\\
{\cal B}_{Yn} \equiv &\sum_{1\le i<j\le n}{a_{ij}^+(b_i+b_j)+a_{ij}^-(b_i-b_j)}\le n(n-1) \label{By}\\
{\cal B}_{Zn} \equiv &\sum_{1\le i<j\le n}{a_{ij}(b_i-b_j)}\le \left\lfloor\frac{n^2}{2}\right\rfloor,
\label{Bz}
\end{align}
\end{subequations}
where $a$ and $b$'s can pick up the values $\pm 1$ and $\lfloor x \rfloor$ denotes the largest integer smaller or equal to $x$. For given $b$'s the above expressions can be maximized by choosing $a$'s with the same signs as for the linear combinations of $b$'s arising in the parentheses. Then one may observe that by taking $a$'s this way the expressions ${\cal B}_{Xn}$ and ${\cal B}_{Yn}$ are invariant under any change of indices or signs of $b$'s. Therefore one may choose for instance $b_i=+1$ for all $1\le i \le n$ saturating the values of ${\cal B}_{Xn}$ and ${\cal B}_{Yn}$ and resulting in the bounds appearing on the left-hand side of Eqs.~(\ref{Bx}-\ref{By}). Further, one may arrive at the bound $\lfloor\frac{n^2}{2}\rfloor$ for ${\cal B}_{Zn}$ by noting that this expression is invariant under any permutation of the indices (such as for ${\cal B}_{Xn}$ and ${\cal B}_{Yn}$) and thus only the overall number of $+1$ and $-1$ values matter. Now suppose that $+1$ occurs $k$ times while $-1$ occurs $n-k$ times among the values of $b$'s, resulting in the sum $2(n-k)k$. This is maximal for $k= \lfloor n/2 \rfloor$ yielding the bound $\lfloor\frac{n^2}{2}\rfloor$ appearing in the left-hand side of Eq.~(\ref{Bz}).
Then averaging these expressions over the ensemble of the runs of the experiment and considering the relationship in Eq.~(\ref{maxBMLHV}), one obtains the following family of Bell inequalities (which we call as $Xn$, $Yn$, and $Zn$), satisfied by any LHV model:
\begin{subequations}
\label{BxyzLHV}
\begin{align}
\langle {\cal B}_{Xn}\rangle &\le \left\lfloor\frac{n+1}{2}\right\rfloor \binom{n}{\left\lfloor\frac{n+1}{2}\right\rfloor} \label{BxLHV}\\
\langle {\cal B}_{Yn}\rangle &\le n(n-1) \label{ByLHV}\\
\langle {\cal B}_{Zn}\rangle &\le \left\lfloor\frac{n^2}{2}\right\rfloor.
\label{BzLHV}
\end{align}
\end{subequations}

Note, that the inequalities $X_2$ and $Y_2$ are in fact the CHSH inequality \cite{CHSH69}, whereas $X_3$ is just a member of the elegant Bell inequalities \cite{BG03} denoted by $S_{3\times4}$ in Ref.~\cite{Gisin07}.
In the next section we focus our attention on the above family of Bell inequalities~(\ref{BxyzLHV}) by the particular value $n=4$ and will find numerically that the maximum quantum value on the Bell inequalities $X_4$ and $Y_4$ cannot be achieved by the use of a pair of qubits.

\section{Numerically obtained quantum bounds}\label{numeric}

We present detailed results concerning the general quantum bound and the quantum bound attainable with qubits on the Bell families (\ref{BxyzLHV}) by $n=4$. While the former bound can be obtained rigorously, we apply numerical techniques to calculate the qubit bound. Therefore, the main findings of this section concerning the power of using higher dimensional quantum systems over qubits in violating Bell inequalities will be based on numerically computed results.

In order to calculate these quantum and qubit bounds, we use their connections with the sum of norms of linear combinations of unit vectors in the Euclidean space, established in Section~\ref{extremal}. In particular, according to Eqs.~(\ref{maxBMQM},\ref{maxBM2D3D}), the general quantum bound, the quantum bound of qubits associated with real and complex numbers are equal to $\max \langle {\cal B}_M\rangle_d$ in Eq.~(\ref{maxBMd}) for $d=m_B$ and $d=2,3$, respectively. By comparing the definition in Eq.~(\ref{BMd}) with the expressions in Eqs.~(\ref{Bxyz}) the $M_{ij}$ coefficients entering in the objective function can be easily extracted.

Thus in this section our aim is to determine the value of
\begin{equation}
\max\langle{\cal B}_M\rangle_{d}=\max\sum_{i=1}^{m_A}\left|\sum_{j=1}^{m_B}M_{ij}\vec b_j\right|
\label{objective}
\end{equation}
with the constraints $\vec b_j \in S^{d-1}$ for $d=m_B, 2, 3$ with the particular $M$-matrices corresponding to the Bell coefficients of the inequalities $X_4$, $Y_4$, and $Z_4$. By a suitable parametrization of the sphere $S^{d-1}$ we can omit the constraints allowing us to use techniques of unconstrained optimization.  The numerical optimization has been performed actually by the aid of the Downhill Simplex method \cite{NM65}. To maximize the possibility that the extremal points found were actual global extremal points, each optimization task was started at least $1000$ times from randomly generated initial points. Furthermore, we verified by using the BARON code \cite{BARON} that the solutions in fact correspond to global maxima. This program is a general purpose solver for global optimization problems \cite{Neumaier04} based on branch and box reduction technique \cite{TS02}.

The results of this extensive optimization have been summarized in Table~\ref{table:XYZ}.
The three separate rows in the table represent results for the Bell inequalities $X_4$, $Y_4$ and $Z_4$ (represented in the first column). The next two columns from left denote the number of measurement settings on Alice and Bob's side, the ensuing four columns in turn show the bounds for the actual Bell inequalities by $d=1,2,3,4$ signifying the bounds achievable by LHV, by qubits associated with real and complex numbers, and with general quantum systems. The last separate column gives the ratio of the violation of the Bell inequalities obtainable by general quantum systems ($d=4$) relative to qubits ($d=3$).

Recall that the LHV bounds ($d=1$) are given by Eqs.~(\ref{Bxyz}) in Sec.~\ref{newfam} and we stress that the general quantum bounds ($d=4$) can also be obtained analytically either by using geometrical considerations or techniques from semidefinite programming such as discussed in Refs.~\cite{Wehner06} and \cite{NPA07}. Hence, in essence only the cases $d=2$ and $d=3$ necessitate numerical optimization. Moreover, as it will turn out in the next section, the Bell inequality $Z_4$ can be treated by analytical means for each cases $d=1,2,3,4$.


\begin{table*}
\begin{center}
\begin{tabular} {|l||l|l||l|l|l|l||l|} \hline
$ $ & $ $ & $ $ & $d=1$ & $d=2$ & $d=3$ & $d=4$ & $(d=4)/(d=3)$\\
$ $ & $m_A$ & $m_B$ & $\text{LHV}$ & $\text{2D}$ & $\text{3D}$ & $\text{QM}$ & $\text{QM/3D}$\\ \hline\hline

$\max \langle {\cal B}_{X4}\rangle_d$ & $8$ & $4$ & $12$ & $14.81$ & $15.45$ & $16$ & $1.036$\\ \hline
$\max \langle {\cal B}_{Y4}\rangle_d$ & $12$ & $4$ & $12$ & $16.109$ & $16.726$ & $16.976$ & $1.015$\\ \hline
$\max \langle {\cal B}_{Z4}\rangle_d$ & $6$ & $4$ & $8$ & $9.657$ & $9.798$ & $9.798$ & $1$\\ \hline
\end{tabular}

\end{center}
\caption{
Classical ($d=1$), qubit associated with real and complex numbers ($d=2$ and $d=3$), and general quantum bounds ($d=4$) for the Bell inequalities $X_4$,$Y_4$ and $Z_4$. The ratio of the bounds for general quantum systems relative to qubits associated with complex numbers are shown in the last column for the three respective Bell inequalities.}
\label{table:XYZ}
\end{table*}

Regarding the case $d=4$ the optimal configuration by the Bell inequalities $X_4$ and $Y_4$ is the one where the vectors $\vec b$ are mutually orthogonal, i.e., lying on the coordinate axes spanning the four-dimensional Euclidean space. Geometrically, the maximization specified by Eq.~(\ref{objective}) corresponds to maximizing the sum of the lengths of all space diagonals (for $d=4$ quadragonals) and of all face diagonals of a four-dimensional rhombohedron in cases of $X_4$ and $Y_4$, respectively. In both cases the optimum shape is the $4$-cube. For $Z_4$ the optimum is reached by $\vec b$'s forming the vertices of the regular tetrahedron inscribed in the unit sphere $S^2$.

The following conclusions can be made from the numbers presented in Table~\ref{table:XYZ}:
Numerically the Bell inequalities $X_4$ and $Y_4$ provide examples that the general bipartite quantum systems ($d=4$) outperform the bound pertaining to the qubit case ($d=3$). Further, $X_4$ is more powerful than $Y_4$ yielding the ratio $\sim 1.036$ of the quantum per qubit bounds. Moreover, by each of the three Bell inequalities measurements on qubits needing complex numbers ($d=3$) give better performance than what can be obtained with measurements on qubits ($d=2$) requiring real numbers. Note, that the elegant construction of Bechmann-Pasquinucci and Gisin \cite{BG03} has the same property, i.e., requires complex numbers to obtain maximal violation of their Bell inequalities.
On the other hand, we haven't found in the literature a Bell inequality which would have been proved to possess the former property, i.e., the quantum per qubit ratio is bigger than 1. However, this does not mean that no such Bell inequalities exist in the literature. One such example is the inequality constructed by Fishburn and Reeds \cite{FR94} which has 20 measurement settings on each side, and it is the simplest known explicit example with the property that the quantum per classical ratio is larger than $\sqrt 2$. This value, 1.4285, is analytically known, and it can be achieved with $d=5$. On the other hand, numerically, we have got the value of 1.3519 (smaller than $\sqrt 2$) for the qubit per classical ratio, entailing the $\sim 1.056$ quantum per qubit ratio.
Due to the many parameters involved it required a much longer calculation than the previous examples, but we are fairly confident that the result is correct. We managed to reproduce the same way the known $d=5$ value as well, although that problem has even more parameters.
Actually in Ref.~\cite{FR94} a family of Bell inequalities were constructed. Although the member of this family with 12 measurement settings on each side has a smaller than $\sqrt 2$ quantum per classical ratio (1.4), it still has the property that it can not be maximally violated with qubits. We have got 1.3485 for the qubit per classical ratio.

To conclude, computations show that in order to maximally violate the Bell inequalities $X_4$ and $Y_4$ one needs to resort to systems of higher than two-dimensional local Hilbert spaces. However, in the next section we will see that the family of Bell inequality $Z_n$ for higher $n$ values (especially for $n=6$) proves to be suitable for demonstrating analytically as well the advantage of higher dimensions systems over qubits.

\section{Analytically derived quantum bounds}\label{analytic}

Let us observe that
\begin{equation}
\label{End}
E(n,d)\equiv \max\sum_{1\le i,j\le n}{|\vec b_i-\vec b_j|}=\max\langle {\cal B}_{Zn}\rangle_d,
\end{equation}
with the unit vectors $\vec b$ in $\RR^d$, is an extremal problem in discrete geometry \cite{AS03}: Dispose $n$ points on the unit sphere of $\RR^d$ on such a way that the sum of distances between all pairs of points would attain its maximum. This problem was originally posed by Fejes T\'oth \cite{Fejes56} some fifty years ago and settled the problem completely for the $d=2$ planar case and also for the case $d=n-1$. There are a number of other instances of $n,d$ where the value of $E(n,d)$ is known either due to exact, analytical treatment \cite{Fejes56,CK07,Alexander72} or by the mean of extensive numerical calculations \cite{RSZ95}. However, for now on we devote our attention only to exact results.

First let us observe that in Eq.~(\ref{End}) in the sum only differences appear entailing $E(n,n-1)=E(n,n)$, that is, it suffices to consider $(n-1)$-dimensional Euclidean space to maximize the sum of distances between $n$ points. This may be proved by noticing that the vectors $\{\vec b_i - \vec b_j\}_{1\le i <j\le n}^n$ span an $(n-1)$-dimensional subspace. By projecting the unit vectors $\vec b$'s on this subspace, as a result these projected vectors are shortened but still have equal length. Then the sum of distances can always be increased by stretching these vectors to unit length in this $(n-1)$-dimensional subspace.

The case $\RR^2$, as already mentioned, is completely solved analytically \cite{Fejes56}, yielding $E(n,2)=n \cot [\pi/(2n)]$, which tends to $2n^2/\pi$ for large $n$ values and the corresponding optimal configurations are the regular $n$-gons.

In case $\RR^3$ there are exact results only either for a small number of points or in the limit of infinitely many points. Especially, Cohn and Kumar \cite{CK07} defined a configuration to be universally optimal if it maximizes the energy expression $\sum_{1\le i<j \le n}f(|\vec b_i-\vec b_j|)$ for all completely monotonic $f$. These configurations are a subset of our optimal configurations, thus of special interest for us. Interestingly, Cohn and Kumar were able to find universal optimality only for certain special arrangements (a list of them can be found in Table~1 of Ref.~\cite{CK07}). In three dimensions, the examples comprise the vertices of a regular tetrahedron, octahedron, or icosahedron giving the respective values $E(4,3)=4 \sqrt 6$, $E(6,3)=6(1+2\sqrt 2)$, and $E(12,3)=12( 1 + \sqrt{5(5+2\sqrt 5)})$. On the other hand, in the asymptotic limit, $E(n,3)=\lim_{n\rightarrow\infty}2n^2/3$. This result is due to Alexander \cite{Alexander72}.

In higher dimensions much less is known, but there exist particular sets of solutions, such as $E(2n,n)=2n(1+(n-1)\sqrt 2)$ with the regular cross polytope as the corresponding optimal configuration \cite{KY97}, and $E(n,n-1)=n\sqrt{n(n-1)/2}$ with the regular simplex \cite{Fejes56} as the optimal configuration.


\begin{table*}
\begin{center}
\begin{tabular} {|l||l|l||l|l|l||l|} \hline
$E(n,d)=$ & $ $ & $ $ & $(d=2)/(d=1)$ & $(d=3)/(d=1)$ & $(d=n-1)/(d=1)$ & $(d=n-1)/(d=3)$ \\
$\max \langle {\cal B}_{Zn}\rangle_d$ & $m_A$ & $m_B$ & $\text{2D}/\text{LHV}$ & $\text{3D}/\text{LHV}$ & $\text{QM}/\text{LHV}$ & $\text{QM}/\text{3D}$ \\ \hline\hline

$n=4$ & $6$ & $4$ & $1.207$ & $1.225$ & $1.225$ & $1$\\ \hline
$n=6$ & $15$ & $6$ & $1.244$ & $1.276$ & $1.291$ & $1.0116$\\ \hline
$n\rightarrow\infty$ & $\binom{n}{2}$ & $n$ & $4/\pi=1.273$ & $4/3=1.333$ & $\sqrt 2 = 1.414$ & $1.0607$\\ \hline
\end{tabular}

\end{center}
\caption{
Analytically obtained, exact results rounded up to three decimals for the $Z_n$ family of Bell inequalities for the values $n=4,6$ and $n\rightarrow\infty$.}
\label{table:Z}
\end{table*}

In Table~\ref{table:Z} we collected analytical values of $E(n,d)=\max\langle {\cal B}_{Zn}\rangle_d$ for $n=4,6$ and for $n\rightarrow\infty$ by the dimensions $d=1,2,3$ and $d=n-1$. The second and third columns, such as in Table~\ref{table:XYZ}, present the number of measurement settings per party, whereas the next three columns show ratios of the bounds relative to the LHV bound achievable by qubits with real numbers ($d=2$), by qubits with complex numbers ($d=3$) and by quantum system disregarding the size of the used Hilbert spaces ($d=n-1$). Finally, the last separate column provides us with the ratio of the violation of the Bell inequalities $Z_n$ obtainable by general quantum systems ($d=n-1$) relative to qubits ($d=3$). Actually, the case $n=4$ is the one discussed before in Sec.~\ref{numeric}, showing that this case can indeed be treated purely analytically.

However, for us the case $n=6$ has particular interest. In the last column for $n=6$ we can read off the value $(\sqrt {120} -\sqrt {15})/7\sim 1.0116$ showing conclusively that for the Bell inequality $Z_6$ the general quantum bound is higher than the quantum bound corresponding to qubits. This result provides us with an answer for Gill's question \cite{Gill}, demonstrating that even in the case of Bell inequalities for two-outcome measurements ($k=2$) and for two parties one sometimes needs to choose local Hilbert spaces of dimension more than $k=2$ to obtain a maximal quantum violation.

When $n\rightarrow \infty$ we may observe that qubits associated with real numbers in the Bell inequalities $Z_n$ are just as powerful as in the Bell inequalities of Gisin \cite{Gisin99} giving the ratio $4/\pi$ for the violation. However in our situation qubits associated with complex numbers can even do better, and finally  the general quantum bound corresponds to the ratio $\sqrt 2$ such as in the CHSH inequality \cite{CHSH69}. Also note that for $n\rightarrow\infty$ the general quantum bound on $Z_n$ outperforms the qubit bound by $\sim6\%$, a value which is larger than $\sim 3.4\%$ corresponding to the Bell inequality $X_4$. In this respect, however, inequality $Z_n$ for $n\rightarrow\infty$ is apparently not amenable to an experimental test in contrast to the inequality $X_4$.

For completeness, we give some details about the observables and states associated with the measurement process to obtain the highest quantum value of the particular Bell inequality $Z_6$. Due to the results of Fejes T\'oth \cite{Fejes56}, the set $\{\vec b_k\}_{k=1}^n$  which attains $E(n-1,n)=\max\langle {\cal B}_{Zn}\rangle_\text{QM}$ corresponds to the vertices of the regular simplex in dimension $n-1$. Given these vectors $\vec b$ Alice forms her vectors $\vec a_{ij}$, $1\le i<j\le n$ by using the formula in Eq.~(\ref{ai}). Then by the sense of Tsirelson's theorem \cite{Tsirelson87} (alternatively Lemma~\ref{lemma:second}), the required local dimension for $n=6$ is $d=2^{\lfloor (n-1)/2\rfloor} =4$ and the corresponding state is $|\Phi^+\rangle = \frac{1}{2}\sum_{i=1}^4|ii\rangle$. Using five mutually anti-commuting operators on $\CC^4$, i.e., the Dirac matrices, it is straightforward to build up the observables $A_{ij}$ and $B_k$:
\begin{subequations}
\label{AB}
\begin{align}
A_{ij} &= \sum_{l=1}^{5}{a_{ij}^{(l)}\gamma_l},\hspace{5 mm} 1\le i,j \le n,\\
B_k &= \sum_{l=1}^{5}{b_k^{(l)}\gamma_l^t},\hspace{5 mm} 1\le k \le n,
\end{align}
\end{subequations}
where $n=6$, $a_{ij}^{(l)}$ ($b_k^{(l)}$) are the components of the vectors $\vec a_{ij}$ ($\vec b_k$) and the five gamma matrices are given by $\{\gamma_l\}_{l=1}^5$ as a special case of the formuale~(\ref{gamma}) in the proof of Lemma~\ref{lemma:second}, as follows
\begin{subequations}
\label{gammaspec}
\begin{align}
\gamma_1 & = \sigma_x\otimes\one \\
\gamma_2 & = \sigma_y\otimes\one \\
\gamma_3 & = \sigma_z\otimes\sigma_x \\
\gamma_4 & = \sigma_z\otimes\sigma_y \\
\gamma_5 & = \sigma_z\otimes\sigma_z.
\end{align}
\end{subequations}
Note, that the observables attaining the general quantum optimum in the case of inequalities $X_4$ and $Y_4$ can also be formed by Eqs.~(\ref{AB}) (and with the same maximally entangled state), but with mutually orthogonal $\vec b$-vectors (and with the corresponding $\vec a$-vectors in Eq.~(\ref{ai})) and by setting $n=4$ in Eqs.~(\ref{AB}).

\section{Discussions}\label{disc}

There are several issues related to the problem in question but not considered in the present treatment.
For instance, it would be interesting to improve further the quantum per qubit ratio from the value $\sim 1.0607$ achieved by the Bell inequality $Z_n$ for $n\rightarrow \infty$. Actually, the quantum per qubit ratio is closely related to the amount of noise which can be mixed to a maximally entangled state, which ruins the nonlocal correlations so that they could be reproduced by a maximally entangled qubit as well.
In particular, a better value (the lower bound $\sim 1.106$) for the quantum per qubit ratio corresponds to the Reeds-Davie Bell Inequality \cite{DavieReeds}. The actual value of this lower bound follows from Ref.~\cite{Brunner} where the gap between the best upper bound on the three-dimensional Grothendieck constant \cite{DavieReeds} and the lower bound on the Grothendieck constant \cite{Krivine79} has been exploited to show the existence of quantum correlations of two outcomes not achievable with qubits.
Ref.~\cite{Brunner} has also given an important motivation for this work and to subsequent research in this direction. In this work the notion of dimension witnesses has been introduced, which makes it possible to measure dimensions of Hilbert spaces. The inequalities found in the present paper may be regarded as examples of dimension witnesses.

In the present treatment we dealt only with correlation Bell inequalities of the CHSH-type without entering marginals in the Bell expressions. This enabled us to use Tsirelson's vectorial formalism considerably simplifying the optimization problem. When marginals are involved one may resort to semidefinite programming tools to obtain upper bounds on the quantum violation of Bell inequalities \cite{NPA07}. Note, however, that attaining tight bounds especially in the local qubit cases (for two-outcome measurements) is still a difficult problem and probably only tractable by numerical optimization methods with resources quickly increasing in the number of measurement settings. Recently, we determined \cite{PV07} numerically the maximum quantum violation in higher dimensional Hilbert spaces of over 100 bipartite Bell inequalities with marginals with up to five measurement settings per party.

The question also naturally arises, if one could somehow establish tight quantum bounds not only for qubits on CHSH-type Bell inequalities, the case which was treated in the present paper, but for higher dimensional quantum systems as well.
We also did not consider Bell inequalities with more than two-outcome ($k>2$) measurements. Namely, one may ask, if it were possible to construct $k$-outcome Bell inequalities with $k>2$, where more than $k$-dimensional local Hilbert spaces are required to attain the maximal quantum violation. Note, that in the case of Collins et al. inequalities \cite{CGLMP} having $k$ outputs per measurement, the maximum quantum value can be already attained on $k$-dimensional local Hilbert spaces up to $k<9$. This fact has been established by Navascu\'es et al. \cite{NPA07}.

Finally, an open question which concerns our present work more closely, whether there exist correlation Bell inequalities which improve on the family of Bell inequalities presented here in the sense that Alice and Bob need to perform less measurements in order to achieve a stronger violation of the LHV bound by qubits than by general quantum systems. For this purpose it would be interesting to present a Bell inequality with less than $m_A+m_B=8+4=12$ settings, the case corresponding to our $X_4$ inequality. On the other hand, it would be useful from the viewpoint of realization to find inequalities where the quantum per qubit ratio could be increased for few settings above $\sim 1.036$, the value we found numerically for the $X_4$ inequality. We wish to address some of these questions in a future paper.

\acknowledgements

T.V. is indebted to Antonio Ac\'in for valuable discussions. The authors are indebted to an anonymous referee for drawing our attention to Bell inequalities appearing in Refs.~\cite{FR94,DavieReeds} relevant to this study and for the helpful comments.

\appendix

\section{Tsirelson's theorem on quantum realizable correlations}\label{app}

We reformulate Tsirelson's original problem about the realization of quantum correlations with dot products by an explicit construction of the representation of the Clifford Algebra. This treatment closely follows the proof of Lemma~2 presented in Ref.~\cite{AGT06}.

Suppose that Alice and Bob measure observables $A_i, i \in \{1,\ldots,m_A\}$ and $B_j, j \in \{1,\ldots,m_B\}$ on a pure quantum state $|\psi\rangle \in \CC^{D}\otimes\CC^{D}$, where $D$ denotes some finite dimension. Assume the condition that $A_i^2 = \one$, $B_j^2 = \one$, which should be satisfied by projective measurements having binary outcomes. Define the joint correlations $\langle \alpha_i \beta_j\rangle_{\psi} = \langle \psi|A_i\otimes B_j|\psi\rangle$. The following Lemma is borrowed from Refs.~\cite{AGT06,RT07} with a small modification:

\begin{lemma}
We can associate real unit vectors $\vec a_i\in\RR^{m_A+m_B}$ with $A_i$, $\vec b_j\in\RR^{m_A+m_B}$ with $B_j$ such that $\langle \alpha_i \beta_j\rangle_{\psi} = \vec{\vphantom{b} a}_i \cdot \vec b_j$, for all $1\le i\le m_A$ and $1\le j\le m_B$.
\label{lemma:first}
\end{lemma}

\begin{proof}
Let $|a_i\rangle = A_i \otimes \one |\psi\rangle$ and $|b_j\rangle = \one \otimes B_j |\psi\rangle$. Then $\langle \alpha_i \beta_j \rangle_{\psi} = \langle a_i|b_j \rangle$ and owing to $A_{i}^2 = B_j^2 = \one$, $\langle a_{i}|a_{i} \rangle = \langle b_{i}|b_{i} \rangle = 1$. Let $a_i^{(k)}\in\CC$ the components of $|a_i\rangle$, and similarly let $b_j^{(k)}\in\CC$ the components of $|b_j\rangle$, where $k=1,2,\ldots,D^2$.
We now define the $(2D^2)$-dimensional real vectors $\vec a_i =( \mathrm{Re}\, a_i^{(1)}, \mathrm{Im}\, a_i^{(1)}, \mathrm{Re}\, a_i^{(2)}, \mathrm{Im}\, a_i^{(2)},\ldots, \mathrm{Re}\, a_i^{(D^2)}, \mathrm{Im}\, a_i^{(D^2)})$, and
$\vec b_j =( \mathrm{Re}\, b_j^{(1)}, \mathrm{Im}\, b_j^{(1)}, \mathrm{Re}\, b_j^{(2)}, \mathrm{Im}\, b_j^{(2)},\ldots, \mathrm{Re}\, b_j^{(D^2)}, \mathrm{Im}\, b_j^{(D^2)})$.
Then $\vec a_i \cdot \vec a_i = \vec b_j \cdot \vec b_j =1$ and consequently $\langle \alpha_i \beta_j \rangle_{\psi} = \vec a_i \cdot \vec b_j $ (because $\langle a_i|b_j\rangle$ is real) for all $1\le i\le m_A$ and $1\le j\le m_B$. Further, since we have $m_A+m_B$ number of vectors entering in the dot products it is sufficient to consider an Euclidean space with dimension $m_A+m_B$ containing vectors $\vec a_i$, $\vec b_j$.
\end{proof}
The converse of Lemma~\ref{lemma:first} is also true:

\begin{lemma}
Let $\{\vec a_i\}_{i=1}^{m_A}$ and $\{\vec b_j\}_{j=1}^{m_B}$ be sets of unit vectors in $\RR^n$. Let ${\rm d}=2^{\lfloor n/2\rfloor}$ and $|\Phi^+\rangle=1/\sqrt{{d}} \sum_{i=1}^{{d}}{|ii\rangle}$. Then there are observables $A_1,\ldots, A_{m_A}$ and $B_1,\ldots, B_{m_B}$ on $\CC^{{d}}$ such that

\begin{align}
\langle \alpha_i\rangle_{\Phi^+} &= \langle \Phi^+|A_i \otimes \one|\Phi^+\rangle = 0 \\
\langle \beta_j\rangle_{\Phi^+} &= \langle \Phi^+|\one \otimes B_j|\Phi^+\rangle = 0 \\
\langle \alpha_i \beta_j\rangle_{\Phi^+} &= \langle \Phi^+|A_i \otimes B_j|\Phi^+\rangle = \vec{\vphantom{b} a}_i \cdot \vec b_j
\label{full}
\end{align}
for all $1\le i\le m_A$ and $1\le j\le m_B$.
\label{lemma:second}
\end{lemma}

\begin{proof}
Let us calculate $\langle \alpha_i \beta_j\rangle_{\Phi^+} = \tr (A_i B_j^t)/{d}$,  where $t$ denotes transposition. Introduce a set of $n$ anti-commuting matrices, $\gamma_i$, $1\le i\le n$ which constitutes a $({d}=2^{\lfloor n/2\rfloor})$-dimensional representation of the Clifford Algebra,
\begin{equation}
\{\gamma_i, \gamma_j\}=2 \delta_{ij}.
\end{equation}
Let us construct this representation by tensor product of $\lfloor n/2\rfloor$ Pauli matrices. Depending on the number $n$, ($n=2k$, $n=2k+1$) we have the following gamma matrices
\begin{align}
\gamma_1 & = \sigma_x\otimes\one\otimes\cdots\otimes\one\otimes\one \nonumber\\
\gamma_2 & = \sigma_y\otimes\one\otimes\cdots\otimes\one\otimes\one \nonumber\\
\gamma_3 & = \sigma_z\otimes\sigma_x\otimes\cdots\otimes\one\otimes\one \nonumber\\
\gamma_4 & = \sigma_z\otimes\sigma_y\otimes\cdots\otimes\one\otimes\one \nonumber\\
         & \vdots \nonumber \\
\gamma_{2k-3} & = \sigma_z\otimes\sigma_z\otimes\cdots\otimes\sigma_x\otimes\one \nonumber\\
\gamma_{2k-2} & = \sigma_z\otimes\sigma_z\otimes\cdots\otimes\sigma_y\otimes\one \nonumber\\
\gamma_{2k-1} & = \sigma_z\otimes\sigma_z\otimes\cdots\otimes\sigma_z\otimes\sigma_x \nonumber\\
\gamma_{2k} & = \sigma_z\otimes\sigma_z\otimes\cdots\otimes\sigma_z\otimes\sigma_y \nonumber\\
\gamma_{2k+1} & = \sigma_z\otimes\sigma_z\otimes\cdots\otimes\sigma_z\otimes\sigma_z.
\label{gamma}
\end{align}
It can be easily checked that any of two gamma matrices anti-commute, while the square of any one is an identity matrix. Thus we can form a basis from them satisfying $\tr(\gamma_i \gamma_j)=d\delta_{ij}$. Let $A_i=\sum_{k=1}^n{a_i^{(k)}\gamma_k}$ and $B_j=\sum_{k=1}^n{b_j^{(k)}\gamma_k^t}$, which define the observables $A_i$ and $B_j$, where $t$ denotes transposition and $a_i^{(k)}$, $b_j^{(k)}$ denote the elements of the unit vectors $\vec a_i$ and $\vec b_j$, respectively. Squaring the above definitions one obtains $A_i^2= \sum_{k=1}^n{\left(a_k^{(i)}\right)^2 \gamma_k^2=\one}$ and $B_j^2= \sum_{k=1}^n{\left(b_k^{(i)}\right)^2 (\gamma_k^t)^2=\one}$. Thus $A_i$ and $B_j$ are indeed observables. Further, $\tr \left(A_i B_j^t\right) = d\sum_{k=1}^n{a_i^{(k)}b_j^{(k)}}$. This implies that $\langle \alpha_i \beta_j \rangle_{\Phi^+} = \sum_{k=1}^n{a_i^{(k)}b_j^{(k)}} = \vec{\vphantom{b} a}_i \cdot \vec b_j$, where $\vec a_i$ and $\vec b_j$ are unit vectors. Finally, $\langle \alpha_i\rangle_{\Phi^+}=\tr A_i =0 $ and $\langle \beta_j\rangle_{\Phi^+}=\tr B_j =0 $ due to the traceless gamma matrices.
\end{proof}

\end{document}